\documentclass[paper]{JFM-FLM_Au}

\usepackage{bm,mathtools,microtype}
\makeatletter
\def\oddabsfooterflag{\hbox to \textwidth{\hfill{\cppagefont\thepage}}}
\def\evenabsfooterflag{\hbox to \textwidth{{\cppagefont\thepage}\hfill}}
\def\pagelimitfooter{\hbox to \textwidth{{\cppagefont\thepage}\hfill}}
\def\ps@headings{\let\@mkboth\markboth
  \def\@oddhead{\hfill{\itshape\@righttitle}\hfill}
  \def\@evenhead{\hfill\itshape\@lefttitle\hfill}
  \def\@oddfoot{\hbox to \textwidth{\hfill{\cppagefont\thepage}}}
  \def\@evenfoot{\hbox to \textwidth{{\cppagefont\thepage}\hfill}}
  \def\sectionmark##1{\markboth{##1}{}}
  \def\subsectionmark##1{\markright{##1}}}
\makeatother
\hypersetup{
  hidelinks,
  pdftitle={Three-dimensional shock-weak-discontinuity interactions: curvature deposition and post-shock flow response},
  pdfauthor={Alexander Omelchenko},
  pdfsubject={Curvature deposition and post-shock flow response at three-dimensional shock-weak-discontinuity junctions}
}

\newcommand{\sjump}[1]{\mathopen{[\![}#1\mathclose{]\!]}}
\newcommand{\cjump}[1]{\mathopen{[\![}#1\mathclose{]\!]}_{\mathcal C}}

\newtheorem{theorem}{Theorem}[section]
\newtheorem{proposition}[theorem]{Proposition}
\newtheorem{lemma}[theorem]{Lemma}
\newtheorem{corollary}[theorem]{Corollary}

\newenvironment{proof}{\par\noindent\textit{Proof.}\ }{\hfill$\square$\par}

\lefttitle{A. Omelchenko}
\righttitle{Curvature deposition at three-dimensional shock junctions}
\title{Three-dimensional shock--weak-discontinuity interactions: curvature deposition and post-shock flow response}
\author{Alexander Omelchenko\aff{1}}
\affiliation{\aff{1}Constructor University Bremen, Campus Ring 1, 28759 Bremen, Germany}
\corresau{Alexander Omelchenko, \email{aomelchenko@constructor.university}}

\begin{document}
\maketitle

\begin{abstract}
A weak acoustic, entropy or vortical sheet can cross a finite-strength three-dimensional shock without producing a kink, yet it can deposit a discontinuity of curvature along the intersection curve.  We derive a local Euler junction law that separates the active normal-plane scattering problem from its three-dimensional geometric realization.  The scattering selects one scalar curvature amplitude, while the orientation and pre-existing shape of the shock determine the full curvature tensor and the downstream flow response.  A controlled pair of saddle-shock interactions has identical incident and outgoing modal amplitudes but produces post-shock pressure-gradient vectors separated by 31.6 degrees and vorticity vectors separated by 25.6 degrees.  For a perfect gas, acoustic incidence also possesses curvature-neutral branches: outgoing acoustic, entropy and shear waves remain finite while the curvature channel cancels.  Convected entropy and in-plane vortical sheets deposit curvature with opposite signs, and their coefficients obey an exact relation imposed by total-enthalpy conservation and acoustic orthogonality; line-tangent vorticity is curvature-transparent.  An expanding spherical shock provides an unsteady example with non-zero shock speed and non-zero intersection-line tracking speed, separating the smooth blast-wave acceleration from the deposited acceleration jump.  Near an umbilic, the same junction selects or rotates the principal-curvature frame.  The resulting law supplies the interface data needed to connect smooth curved-shock reconstructions across a continuously differentiable, piecewise twice-differentiable front.
\end{abstract}

\section{Introduction}
\label{sec:introduction}

When a weak acoustic, entropy or vortical sheet crosses a curved three-dimensional shock, it both scatters into outgoing acoustic, entropy and vortical modes and changes the shock surface itself.  The shock position and tangent plane may remain continuous along the spatial intersection curve while the one-sided curvature tensor changes abruptly.  This situation arises naturally when acoustic fronts, entropy spots or vortical sheets encounter compression shocks in three-dimensional inlets, detached bow shocks or expanding blast fronts.  It is also directly relevant to post-shock reconstruction, because the local shock shape enters the determination of pressure, velocity and vorticity gradients immediately behind a curved shock.  A normal-plane interaction calculation may determine all active modal amplitudes and yet remain unable to determine this three-dimensional flow response, because it contains no information about the curvature of the shock along the intersection curve.

Existing theories address the two parts of this problem separately.  Classical plane-shock and locally planar interaction analyses resolve the conversion of incident disturbances into acoustic, entropy and shear modes, including oblique incidence and fully three-dimensional freestream perturbations \citep{Moore1954,Ribner1954,McKenzieWestphal1968,DuckLasseigneHussaini1995,DuckLasseigneHussaini1997,HeSu2026}.  Smooth shock-change and curved-shock theories instead reconstruct post-shock derivatives from a prescribed differentiable front, and now include fully three-dimensional convex, concave and saddle-shaped shocks \citep{Molder2016,Emanuel2018,Radulescu2020,ShiEtAl2020,EmanuelMolder2022,ZhangEtAl2025}.  Neither description supplies the missing interface condition when a weak sheet leaves the shock tangent plane continuous but creates a jump in the second fundamental form.  The resulting shock is continuously differentiable and piecewise twice differentiable, so the two smooth curved-shock reconstructions require a junction law for their one-sided shape operators.

Two classical traditions motivate such a law.  The geometrical shock-dynamics constructions of \citet{Chester1954}, \citet{Chisnell1957}, \citet{Whitham1958} and \citet{Lick1966} provide scalar evolution rules for shock strength and front motion in non-uniform and quasi-one-dimensional flows.  In parallel, the discontinuity-wave theory of \citet{Jeffrey1973,Jeffrey1974}, \citet{BoillatRuggeri1979}, \citet{Ruggeri1980}, \citet{RadhaSharmaJeffrey1993} and \citet{MentrelliEtAl2008} determines reflected and transmitted weak discontinuities at a shock, with the discontinuity of shock acceleration as an additional scalar amplitude.  Related extensions treated elastic and magnetofluid media \citep{Brun1975,Morro1978,Morro1980}.  The author's earlier generalized Chester--Whitham invariant and derivative relations \citep{Omelchenko2001,Omelchenko2002} arose from the same compatibility problem.  The multidimensional step taken here is to identify the corresponding scalar as the amplitude of a rank-one junction of the full shock shape operator and to determine how the embedding of the interaction curve converts it into intrinsic geometry and downstream flow structure.

The construction combines second-order spacetime compatibility with differentiated Rankine--Hugoniot conditions.  At the two shock traces, the differentiated three-dimensional Euler operator splits exactly into an active four-variable subsystem in the plane normal to the interaction curve and a passive line-tangent shear channel.  The active boundary problem selects a scalar response \(\kappa\), and the physical shock reconstructs it as
\begin{equation}
  \cjump{S}=\kappa\,\eta_s\otimes\eta_s^\flat,
  \qquad
  \cjump{K_G}=k_{\mathcal C}\kappa.
  \label{eq:introduction-main-law}
\end{equation}
Here \(S\) is the shape operator, \(K_G=\det S\) is the Gaussian curvature, \(\eta_s\) is the direction in the shock surface normal to the interaction curve, and \(k_{\mathcal C}=\langle S^-\boldsymbol{t},\boldsymbol{t}\rangle\) is the pre-existing normal curvature along that curve.  Thus the gas dynamics selects one scalar amplitude, whereas the three-dimensional embedding determines what that amplitude does to the physical front.  Figure~\ref{fig:conceptual-junction} shows a schematic \(C^1\), piecewise-\(C^2\)
shock junction and the adapted frame used to separate the active normal-plane
response from the line-tangent direction.

\begin{figure}
  \centering
  \includegraphics[width=0.92\columnwidth]{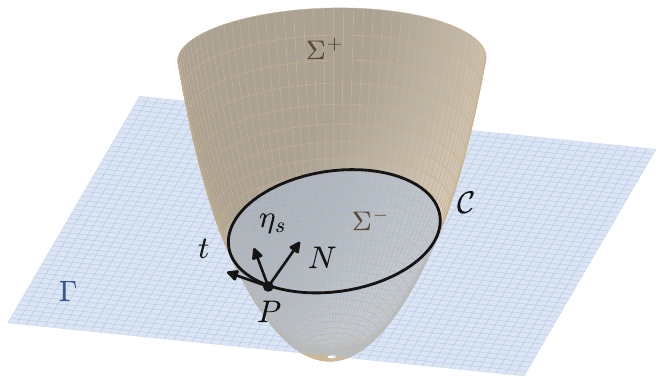}
  \caption{Local geometry of a three-dimensional shock--weak-discontinuity junction.  A weak sheet \(\Gamma\) crosses the shock surface \(\Sigma\) along the interaction curve \(\mathcal C\), dividing it into the two smooth pieces \(\Sigma^-\) and \(\Sigma^+\), which share a common tangent plane along \(\mathcal C\) but have different cross-curve curvatures.  At the marked event \(P\), \(N\) is the shock normal, \(\boldsymbol{t}\) is tangent to \(\mathcal C\), and \(\eta_s=N\times\boldsymbol{t}\) lies in the shock surface, is normal to the interaction curve, and is oriented from \(\Sigma^-\) towards \(\Sigma^+\).  The active Euler problem is posed in \(\operatorname{span}\{N,\eta_s\}\), while line-tangent shear is passive.}
  \label{fig:conceptual-junction}
\end{figure}

The orientation dependence has an observable gas-dynamic consequence.  We construct two local interactions on the same inlet-inspired saddle shock whose complete active normal-plane boundary problems are identical.  They have the same incident trace, outgoing modal amplitudes and scalar curvature response, but the interaction curves have different embeddings in the shock.  Consequently their Gaussian-curvature changes differ, and a one-sided curved-shock reconstruction gives post-interaction pressure-gradient vectors separated by \(31.6^\circ\) and vorticity vectors separated by \(25.6^\circ\).  The comparison isolates a mechanism that cannot be recovered from modal scattering data alone: identical active scattering can produce different three-dimensional post-shock gradient and rotational structures.

The perfect-gas coefficients reveal two further effects.  Acoustic incidence possesses curvature-neutral branches on which the outgoing acoustic, entropy and in-plane shear waves remain non-zero but the geometric forcing cancels.  Convected entropy and in-plane vortical sheets instead deposit curvature throughout the projected-supersonic sector with opposite signs.  Their coefficients satisfy an exact identity fixed by conservation of total enthalpy and by the orthogonality of the outgoing acoustic covector to a downstream convective mode; a purely line-tangent vortical sheet excites only the passive channel and leaves the shock curvature unchanged.  These results connect the junction law directly to elementary constituents of shock--turbulence interaction.

The spacetime formulation also distinguishes the motion of the shock from the motion of the intersection curve across it.  An expanding spherical shock supplies an example with both non-zero normal shock speed and non-zero tracking speed, so that the smooth blast-wave acceleration can be separated from the acceleration jump deposited by the weak sheet.  Near an umbilic shock point, the same rank-one update removes the curvature degeneracy and selects a principal frame; near, but not at, an umbilic, a response comparable with the small principal-curvature gap can rotate the axes by a finite angle without any singularity of the Euler coefficient.

The theory is local and inviscid.  It supplies a first-derivative interface law and the geometric input needed to connect smooth one-sided curved-shock descriptions; it does not replace a global inlet, blunt-body or blast-wave solution.  Section~\ref{sec:local-euler-junction-law} derives the local Euler junction law and the active--passive selection rule.  Section~\ref{sec:steady-euler-curvature-transfer} evaluates the perfect-gas coefficients for acoustic, entropy and vortical incidence.  Section~\ref{sec:three-dimensional-inlet-response} reconstructs the three-dimensional shock and downstream flow response.  Section~\ref{sec:moving-interaction-curves} treats moving spherical shocks and near-umbilic frame selection.  The detailed conormal, Galilean and shock-polar algebra is collected in the supplementary material.

\section{Local Euler junction law}
\label{sec:local-euler-junction-law}
\label{sec:spacetime-transmission}
\label{sec:spacetime-euler-factorization}

This section isolates the local law that converts an incident weak-sheet trace
into a change of shock geometry.  The argument has a simple physical
structure.  Compatibility first shows that a continuously differentiable
shock can acquire only one independent second-derivative jump at the
intersection curve.  The differentiated Rankine--Hugoniot conditions then
select the amplitude of that jump.  Finally, the three-dimensional Euler
operator separates into an active normal-plane problem and a passive shear
channel along the interaction curve.  The full conormal calculation and the
row-reduction algebra are given in the supplementary material; here we retain
the statements needed for the gas-dynamic developments below.

Let the shock history be an oriented world hypersurface \(\Sigma\), with
upstream and downstream Euler states \(U_1\) and \(U_2\), and let a weak
characteristic sheet \(\Gamma\) meet it transversally along the codimension-two
world sheet \(\mathcal C\).  The instantaneous spatial section of this
configuration is shown schematically in figure~\ref{fig:conceptual-junction}.
We use \(\sjump{\cdot}_\Sigma\) for the upstream-to-downstream jump across the
strong shock and \(\cjump{\cdot}\) for the jump between the two smooth shock
pieces \(\Sigma^-\) and \(\Sigma^+\) that meet along \(\mathcal C\).  The state is continuous across \(\Gamma\), although
its first derivatives may jump.  The shock itself is assumed \(C^1\) and
piecewise \(C^2\): its position and tangent hyperplane are continuous, but its
curvature need not be.

Let \(\nu\) be the common shock conormal and
\[
  \mathcal B_\nu(X,Y)=-(D_X\nu)(Y),
  \qquad X,Y\in T\Sigma,
\]
its conormal second fundamental form.  Choose a tangent direction
\(\eta\in T\Sigma\) crossing \(\mathcal C\) within the shock world
hypersurface, and let \(\vartheta\in T^*\Sigma\) annihilate
\(T\mathcal C\) and satisfy \(\vartheta(\eta)=1\).

\begin{proposition}[Second-order compatibility at the junction]
\label{prop:spacetime-rank-one}
There is a scalar \(\alpha\) such that
\begin{equation}
  \boxed{\cjump{\mathcal B_\nu}=\alpha\,\vartheta\otimes\vartheta.}
  \label{eq:spacetime-rank-one}
\end{equation}
Hence the curvature/kinematic junction has rank at most one.
\end{proposition}

\begin{proof}
The common \(C^1\) conormal has identical one-sided derivatives in every
direction tangent to \(\mathcal C\).  Thus
\(\cjump{\mathcal B_\nu}(X,\cdot)=0\) for
\(X\in T\mathcal C\).  Since \(T\Sigma/T\mathcal C\) is one-dimensional and
\(\mathcal B_\nu\) is symmetric, only the transverse component remains,
giving \eqref{eq:spacetime-rank-one}.
\end{proof}

The scalar \(\alpha\) is the single geometric degree of freedom left by
compatibility.  With unit spatial normals on a steady shock it is the
cross-track curvature jump \(\kappa\).  For a moving interaction it also
generates the compatible jumps of normal shock-speed gradient and intrinsic
normal acceleration derived in \S\ref{sec:moving-interaction-curves}.

Write the spacetime Euler flux contracted with a covector \(\xi\) as
\(\mathfrak F_\xi(U)\), and its state Jacobian as
\(\mathcal A_\xi(U)=D_U\mathfrak F_\xi(U)\).  The strong shock satisfies
\[
  \sjump{\mathfrak F_\nu}_\Sigma=0.
\]
Differentiating this condition in the direction \(\eta\) on the two smooth
shock pieces and subtracting gives
\begin{equation}
  \boxed{
  \mathcal A_\nu(U_2)\delta q_2
  -\mathcal A_\nu(U_1)\delta q_1
  +h_\alpha\alpha=0,}
  \qquad
  h_\alpha=-\sjump{\mathfrak F_{\widetilde\vartheta}}_\Sigma .
  \label{eq:local-differentiated-rh}
\end{equation}
Here \(\delta q_i\) is the change of the derivative trace in state \(i\)
when the weak sheet passes the interaction curve, and
\(\widetilde\vartheta\) is any ambient extension of \(\vartheta\).  The
vector \(h_\alpha\) is the shock-geometry sensitivity column: it is the
change of the Rankine--Hugoniot residual produced by a unit change of shock
conormal.  Its definition is independent of the chosen extension because the
zeroth-order shock already satisfies the Rankine--Hugoniot conditions.

Resolve \(\delta q_1\) and \(\delta q_2\) into the prescribed incident weak
mode and the admissible outgoing acoustic, entropy and shear modes.  Their
boundary images form the finite-dimensional system
\begin{equation}
  \boxed{
  B_{\rm out}a_{\rm out}+h_\alpha\alpha
  =b_{\rm in}a_{\rm in}.}
  \label{eq:local-euler-boundary-system}
\end{equation}
Each column of \(B_{\rm out}\) is the residual generated in
\eqref{eq:local-differentiated-rh} by a unit outgoing characteristic trace;
\(b_{\rm in}\) is the corresponding image of the normalised incident trace.
The usual non-glancing and constant-multiplicity assumptions ensure that the
mode count and these trace bases are well defined at the event.

In the scalar regime used below, \(B_{\rm out}\) has a one-dimensional left
nullspace.  Let \(\psi_{\rm loc}\) be any non-zero covector satisfying
\(\psi_{\rm loc}^{\mathsf T}B_{\rm out}=0\).  It filters out every response
that can be supplied by outgoing waves and leaves the one boundary channel
that must be balanced by shock geometry.  Applying it to
\eqref{eq:local-euler-boundary-system} gives
\begin{equation}
  \boxed{
  \alpha
  =K_{\rm loc}a_{\rm in},
  \qquad
  K_{\rm loc}
  =\frac{\psi_{\rm loc}^{\mathsf T}b_{\rm in}}
         {\psi_{\rm loc}^{\mathsf T}h_\alpha}.}
  \label{eq:spacetime-scalar-law}
\end{equation}
The denominator is the local transversality pairing: it is non-zero exactly
when the geometric column is not reproducible by the outgoing modal columns.
Equation~\eqref{eq:spacetime-scalar-law} is therefore a local solvability law,
not by itself a frequency-dependent stability criterion.

We now specialise the operator to three-dimensional Euler flow.  At a fixed
interaction event choose the orthonormal frame
\((N,\eta_s,\boldsymbol t)\), where \(N\) is the shock normal,
\(\boldsymbol t\) is tangent to the spatial interaction curve and
\(\eta_s=N\times\boldsymbol t\).  The active normal plane is
\(\Pi=\operatorname{span}\{N,\eta_s\}\).  Write the primitive state as
\[
  Y=(\rho,a,b,\mathfrak s,w)^{\mathsf T},
  \qquad
  \boldsymbol u=aN+b\eta_s+w\boldsymbol t,
\]
where \(a\) is the normal velocity component and \(c\) denotes the sound
speed.  Because the weak discontinuity has a continuous zeroth-order state
trace, each upstream and downstream state approaches the same limit on the
two shock pieces.  For a genuine shock with non-zero mass flux and pressure
jump, the zeroth-order momentum and mass conditions then fix the same
oriented normal and normal shock speed on both sides of \(\mathcal C\); no
kink is introduced.

Tangential momentum conservation gives the familiar continuity of tangential
velocity,
\begin{equation}
  b_1=b_2=:b_\Sigma,
  \qquad
  w_1=w_2=:w_\Sigma.
  \label{eq:common-tangential-velocity-compact}
\end{equation}
The two components have different roles.  The direction \(\eta_s\) lies in
the active characteristic plane, so \(b_\Sigma\) enters the Doppler factors
and the acoustic mode geometry.  The direction \(\boldsymbol t\) is tangent
to every participating hypersurface at the event; \(w_\Sigma\) is only a
common line-tangent drift.

An inertial Galilean subtraction by \(w_\Sigma\boldsymbol t\), followed by an
invertible conservative row transformation \(\mathcal R_\Sigma\), makes this
decoupling exact at the two shock traces.  Put
\(\widehat Y=(\rho,a,b,\mathfrak s)^{\mathsf T}\) and
\(\mu_\xi=\xi_0+av_N+bv_\eta\) for
\(\xi=\xi_0dt+v_NN^\flat+v_\eta\eta_s^\flat\).

\begin{theorem}[Active normal-plane Euler block]
\label{thm:spacetime-four-plus-one}
At either finite-strength shock trace,
\begin{equation}
  \boxed{
  \mathcal R_\Sigma\mathcal A_\xi(Y_i)
  =
  \begin{pmatrix}
    \widehat{\mathcal A}_\xi(\widehat Y_i)&0\\
    0&\rho_i\mu_{\xi,i}
  \end{pmatrix},
  \qquad i=1,2,}
  \label{eq:spacetime-euler-block-symbol}
\end{equation}
where \(\widehat{\mathcal A}_\xi\) is the full Euler symbol in the normal
plane \(\Pi\).  Moreover,
\(\mathcal R_\Sigma h_\alpha=(\widehat h_\alpha,0)^{\mathsf T}\).
Thus the differentiated shock problem is the direct sum of an active
four-variable Euler boundary problem and a passive scalar equation for
line-tangent shear.
\end{theorem}

\begin{proof}
In the frame moving with \(w_\Sigma\boldsymbol t\), the relative
line-tangent kinetic energy is quadratic in \(w-w_\Sigma\).  Its first
derivative vanishes at the two base traces, while the line-tangent momentum
equation linearises to the scalar factor \(\rho_i\mu_{\xi,i}\).  The geometric
column has no line-tangent component because the interaction conormals
annihilate \((0,\boldsymbol t)\).  The componentwise calculation is recorded
in the supplementary material.
\end{proof}

The block decomposition yields the following gas-dynamic selection rule.

\begin{corollary}[Active--passive selection]
\label{cor:active-passive-selection}
Whenever the prescribed active and passive boundary blocks are invertible,
(i) acoustic, entropy and in-plane vortical incidence cannot generate an
outgoing \(\boldsymbol t\)-polarised shear mode; (ii) a purely
\(\boldsymbol t\)-polarised vortical trace produces no active acoustic,
entropy or in-plane shear response and gives \(\alpha=0\); and (iii) all
active amplitudes, including the curvature or acceleration response, are
independent of the common drift \(w_\Sigma\).
\end{corollary}

The reduction is pointwise: the frame, the states and \(w_\Sigma\) may vary
along \(\mathcal C\), but their smooth longitudinal derivatives do not enter
the frozen first-jet jump matrix.  The theorem therefore does not reduce the
physical shock geometry to two dimensions.  It says only that the local gas
dynamics selects the scalar \(\alpha\) in the normal plane.  The orientation
of \((\eta_s,\boldsymbol t)\) and the pre-existing shape operator then convert
that scalar into the genuinely three-dimensional curvature and post-shock
flow responses studied below.

\section{Perfect-gas response coefficients}
\label{sec:steady-euler-curvature-transfer}

Section~\ref{sec:local-euler-junction-law} reduces the local interaction to one
scalar amplitude, \(\kappa=K_{\rm loc}a_{\rm in}\), selected by the active
Euler problem in the plane
\(\Pi=\operatorname{span}\{N,\eta_s\}\).  We now evaluate that coefficient
for a steady event, or equivalently for a moving event viewed in the
instantaneous shock--intersection frame.  Write
\[
  \bm u_{i\Pi}=a_iN+b\eta_s,
  \qquad
  V_{\Pi i}=\sqrt{a_i^2+b^2},
  \qquad
  M_{\Pi i}=V_{\Pi i}/c_i,
\]
where \(a_i\) and \(b\) are relative velocities when the shock or the
intersection curve moves.  The scalar closure used below applies in the
swept-shock sector
\begin{equation}
  a_1>c_1,\qquad 0<a_2<c_2,\qquad V_{\Pi2}>c_2.
  \label{eq:steady-euler-mode-regime}
\end{equation}
The normal component is subsonic downstream, but the tangential sweep makes
the projected downstream flow supersonic.  Entropy and the two shear modes are
then convected away from the shock, and exactly one downstream acoustic mode
is outgoing.  This leaves one compatibility condition for \(\kappa\).

\paragraph{Outgoing acoustic compatibility and the shock-polar tangent.}
Let \(\theta\) denote the velocity angle in one fixed laboratory basis and
set \(Q=\sqrt{M_\Pi^2-1}\).  The two steady acoustic compatibility covectors
may be written
\begin{equation}
  \omega_\tau=d\theta
  +\tau\frac{Q}{\rho V_\Pi^2}\,dp,
  \qquad \tau=\pm1.
  \label{eq:steady-euler-acoustic-covectors}
\end{equation}
For the orientation used here, the outgoing branch is
\(\chi=\operatorname{sgn}b\), so that the four-dimensional outgoing state
space is \(\ker\omega_\chi\).  A left annihilator of the outgoing boundary
images is therefore
\begin{equation}
  \psi_{\rm loc}^{\mathsf T}
  =\omega_\chi\mathcal A_N(Y_2)^{-1}.
  \label{eq:steady-euler-boundary-cokernel}
\end{equation}
The mode-count proof and the construction of
\eqref{eq:steady-euler-acoustic-covectors} are given in the supplementary
material.

Hold the physical upstream state fixed and rotate the shock normal through a
signed angle \(\sigma\), with \(N_{,\sigma}(0)=\eta_s\).  The corresponding
Rankine--Hugoniot branch is \(Y_2(\sigma)\), and
\begin{equation}
  \mathcal A_N(Y_2)Y_{2,\sigma}=h_\kappa.
  \label{eq:steady-euler-branch-column}
\end{equation}
Thus the geometric sensitivity column is the normal-flux image of the tangent
to the finite-strength shock branch.  Its projection to
\((\Theta,J)\), where \(\Theta\) is the flow-deflection angle and
\(J=p_2/p_1\), is the tangent to the conventional pressure--deflection shock
polar.  All velocity angles are measured in the fixed laboratory basis, not
in a basis co-rotating with the normal.  For the reference state
\(\gamma=7/5\), \(M_{n1}^2=5\) and \(B^2=10/3\), this convention gives the
useful check \(\theta_{2,\sigma}=4/7\).

\begin{proposition}[Shock-polar response coefficient]
\label{prop:steady-euler-curvature-coefficient}
Define the static transversality pairing
\begin{equation}
  \Delta_{\rm loc}
  :=\psi_{\rm loc}^{\mathsf T}h_\kappa
  =\omega_\chi(Y_{2,\sigma}).
  \label{eq:steady-euler-shock-branch-pairing}
\end{equation}
For a normalized upstream incident trace
\(Q_{\rm in}=a_{\rm in}q_{\rm in}\),
\begin{equation}
  \boxed{
  \kappa=K_{\rm loc}a_{\rm in},\qquad
  K_{\rm loc}=
  \frac{\omega_\chi\!\left(
  \mathcal A_N(Y_2)^{-1}\mathcal A_N(Y_1)q_{\rm in}\right)}
       {\omega_\chi(Y_{2,\sigma})}.}
  \label{eq:steady-euler-upstream-curvature-law}
\end{equation}
For a weak sheet overtaking the shock from state 2,
\begin{equation}
  \boxed{
  \kappa=-\frac{\omega_\chi(q_{\rm in})}
                 {\omega_\chi(Y_{2,\sigma})}\,a_{\rm in}.}
  \label{eq:steady-euler-overtaking-curvature-law}
\end{equation}
\end{proposition}

The denominator compares the shock-polar tangent with the outgoing acoustic
compatibility hyperplane:
\begin{equation}
  \Delta_{\rm loc}
  =\theta_{2,\sigma}
   +\chi\frac{\sqrt{M_{\Pi2}^2-1}}
                  {\rho_2V_{\Pi2}^2}\,p_{2,\sigma}.
  \label{eq:steady-euler-delta-polar-form}
\end{equation}
Equivalently, the augmented frozen boundary determinant equals
\begin{equation}
  \mathcal D_{\rm stat}
  =C_{\rm norm}\,\Delta_{\rm loc},
  \qquad C_{\rm norm}\neq0,
  \label{eq:static-lopatinski-identity}
\end{equation}
where \(C_{\rm norm}\) contains the modal basis, incidence factors and volume
form.  Hence \(\Delta_{\rm loc}\neq0\) is the static evolutionary
transversality condition associated with this junction problem.  Abstract
non-vanishing belongs to the classical Lopatinski theory of perfect-gas
shocks \citep{Majda1983,Metivier2001}; the present formulation identifies the
specific pairing entering curvature deposition and evaluates it directly on
the swept shock polar.  No frequency-dependent stability statement is made.

For a calorically perfect gas, introduce
\[
  M_{n1}=a_1/c_1,\qquad B=|b|/c_1,
  \qquad
  R_\gamma(x)=\frac{(x-1)((\gamma-1)x+2)}{(\gamma+1)x}.
\]
The projected-supersonic condition is
\begin{equation}
  B>B_{\rm crit}(M_{n1}),
  \qquad
  B_{\rm crit}(M_{n1})=\sqrt{R_\gamma(M_{n1}^2)}.
  \label{eq:steady-euler-sweep-threshold}
\end{equation}
Exact differentiation of the normal-shock relations shows that both terms in
\eqref{eq:steady-euler-delta-polar-form} have the same sign.  Thus
\begin{equation}
  \Delta_{\rm loc}>0
  \quad\hbox{throughout \eqref{eq:steady-euler-mode-regime}},
  \label{eq:steady-euler-perfect-gas-nonvanishing}
\end{equation}
and, for each fixed finite-strength shock,
\begin{equation}
  \boxed{
  \lim_{M_{\Pi2}\downarrow1}\Delta_{\rm loc}
  =\frac{2c_1^2R_\gamma(M_{n1}^2)}
         {(\gamma+1)M_{n1}^2c_2^2}>0.}
  \label{eq:steady-euler-sonic-limit}
\end{equation}
The scalar regime therefore ends because the two stationary acoustic
directions coalesce, not because the static pairing vanishes.  The limit is
not uniform in the simultaneous weak-shock limit
\(M_{n1}\downarrow1\).  The explicit algebra and the relation to the full
augmented determinant are recorded in the supplementary material.

\paragraph{Acoustic incidence and curvature-neutral scattering.}
Figure~\ref{fig:euler-response-atlas} maps
\(K_{\rm loc}(M_{n1},B)\) for the two incident acoustic Mach-line families,
using \(\gamma=1.4\), \(\rho_1=c_1=1\) and the normalization
\(\delta p_{\rm in}=1\), with the incidence factor absorbed into
\(a_{\rm in}\).  The two panels use separate colour scales because their
ranges differ substantially.  The grey sector has \(M_{\Pi2}\le1\) and lies
outside the real stationary acoustic closure.

\begin{figure}
  \centering
  \includegraphics[width=\textwidth]{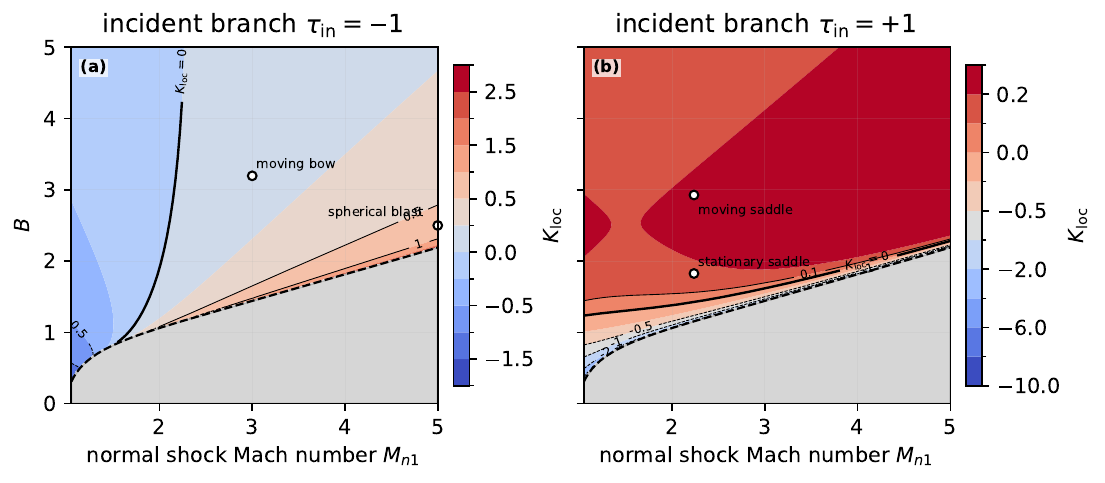}
  \caption{Perfect-gas acoustic curvature coefficient for the two incident
  Mach-line families: (a) \(\tau_{\rm in}=-1\) and
  (b) \(\tau_{\rm in}=+1\).  The vertical coordinate is
  \(B=|\bar b|/c_1\); the dashed curve is \(M_{\Pi2}=1\).  Separate colour
  scales are used for the two branches.  Solid zero contours are
  curvature-neutral incidences: outgoing waves remain non-zero, but the
  geometric channel cancels.  Open circles mark the reference states used
  later.}
  \label{fig:euler-response-atlas}
\end{figure}

Because \(\Delta_{\rm loc}>0\), an acoustic junction is curvature-neutral
precisely when the numerator in
\eqref{eq:steady-euler-upstream-curvature-law} vanishes.  The transmitted
trace then already belongs to the outgoing modal hyperplane, so the boundary
conditions are satisfied by outgoing waves alone and \(\kappa=0\).  This is
neither a sonic singularity nor absence of scattering.  At
\(M_{n1}=2\), for example, neutral points occur at
\(B=1.899770\) and \(B=1.379314\) on the two branches; the outgoing acoustic,
entropy and in-plane shear amplitudes are all non-zero.  The complete modal
values are listed in the supplementary material.

\paragraph{Convected entropy and vortical incidence.}
In the eventwise relative frame, entropy and vortical sheets share the
convective conormal
\[
  \widehat\zeta_{\rm c}
  =\frac{-\bar bN+\bar a_1\eta_s}{V_{\Pi1}}.
\]
With primitive variables \(Y=(\rho,a,b,p,w)^{\mathsf T}\), use
\begin{equation}
\begin{aligned}
  q_{\rm ent}&=(-\rho_1,0,0,0,0)^{\mathsf T},
  &&\delta\mathfrak s/c_p=1,\\
  q_{\rm vort}&=\left(0,c_1\frac{\bar a_1}{V_{\Pi1}},
      c_1\frac{\bar b}{V_{\Pi1}},0,0\right)^{\mathsf T},
  &&|\delta\bm u_\Pi|/c_1=1,\\
  q_t&=(0,0,0,0,c_1)^{\mathsf T}.
\end{aligned}
  \label{eq:nonacoustic-polarizations}
\end{equation}
The associated scalar amplitudes represent entropy-gradient and
velocity-gradient traces, not finite jumps of the continuous state.

\begin{proposition}[Total-enthalpy identity for convected incidence]
\label{prop:nonacoustic-enthalpy-identity}
For the normalizations in \eqref{eq:nonacoustic-polarizations},
\begin{equation}
  \boxed{
  K_{\rm vort}=-\frac{2}{M_{\Pi1}}K_{\rm ent},
  \qquad K_t=0.}
  \label{eq:nonacoustic-enthalpy-identity}
\end{equation}
\end{proposition}

\begin{proof}
Let \(\bar V_{\Pi1}=c_1M_{\Pi1}\) and let \(h_0\) be the total enthalpy per
unit mass, which is continuous through an adiabatic Euler shock.  In the
active variables \((\rho,a,b,p)\), define
\[
 v=(1,\bar a_1,\bar b,\bar V_{\Pi1}^2/2)^{\mathsf T},\qquad
 u=(1,2\bar a_1,2\bar b,h_0+\bar V_{\Pi1}^2)^{\mathsf T}.
\]
The upstream normal-flux images satisfy
\[
 \widehat{\mathcal A}_N\widehat q_{\rm ent}
   =-\rho_1\bar a_1v,
 \qquad
 \widehat{\mathcal A}_N\widehat q_{\rm vort}
   =c_1\rho_1\bar a_1u/\bar V_{\Pi1}.
\]
Now \(u-2v=(-1,0,0,h_0)^{\mathsf T}\).  Total-enthalpy continuity gives
\[
 \widehat{\mathcal A}_N(\widehat Y_2)^{-1}(u-2v)
 =\left(-\frac2{\bar a_2},\frac1{\rho_2},
   \frac{\bar b}{\rho_2\bar a_2},0\right)^{\mathsf T},
\]
and \(\omega_\chi\) annihilates this convective state.  Hence the acoustic
projections of \(u\) and \(2v\) coincide, which yields
\(K_{\rm vort}/K_{\rm ent}=-2c_1/\bar V_{\Pi1}\).  The result
\(K_t=0\) is the passive-channel selection rule of
Section~\ref{sec:local-euler-junction-law}.
\end{proof}

For a perfect gas, set
\(\widehat V_{\Pi2}=V_{\Pi2}/c_1\).  The coefficients may be written compactly as
\begin{equation}
\begin{aligned}
 K_{\rm ent}
 &=-\frac{\chi}{\Delta_{\rm loc}}
 \frac{2}{(\gamma+1)\widehat V_{\Pi2}^{\,2}}
 \left[
 \frac{B}{\sqrt{x}}+
 \frac{(\gamma-1)x+2}{\gamma+1}
 \sqrt{M_{\Pi2}^2-1}
 \right],\\
 K_{\rm vort}&=-\frac{2}{M_{\Pi1}}K_{\rm ent},
 \qquad K_t=0,
 \qquad x=M_{n1}^2.
\end{aligned}
  \label{eq:entropy-forcing-numerator}
\end{equation}
Thus, for \(\bar b>0\),
\(K_{\rm ent}<0<K_{\rm vort}\) throughout the projected-supersonic sector:
positive entropy and in-plane vortical traces bend the shock in opposite
senses, whereas line-tangent vorticity is curvature-transparent.

\begin{figure}
  \centering
  \includegraphics[width=\textwidth]{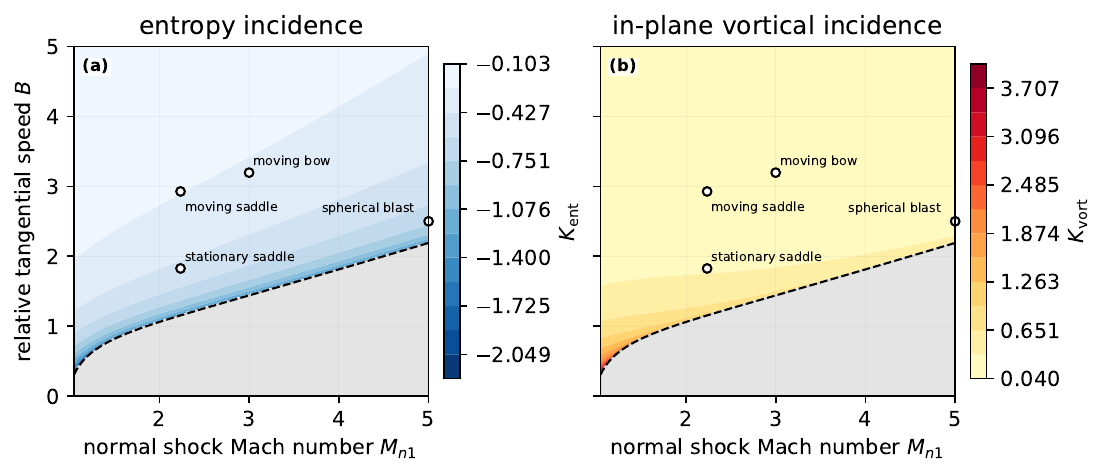}
  \caption{Perfect-gas curvature coefficients for (a) entropy and
  (b) in-plane vortical incidence, in the normalizations of
  \eqref{eq:nonacoustic-polarizations}.  Each panel has its own colour scale;
  the dashed curve is \(M_{\Pi2}=1\).  The pointwise identity
  \(K_{\rm vort}=-2K_{\rm ent}/M_{\Pi1}\) holds throughout the displayed
  sector.}
  \label{fig:nonacoustic-response-atlas}
\end{figure}

Both convected incidences generate outgoing acoustic, entropy and in-plane
shear waves while depositing non-zero curvature.  For the expanding spherical
reference state used later,
\[
  K_{\rm ent}=-\frac57,
  \qquad
  K_{\rm vort}=\frac{4}{7\sqrt5},
  \qquad K_t=0.
\]
Detailed modal amplitudes, reference-state tables and the full algebra behind
the two response maps are provided in the supplementary material.  The same
material also gives an independent planar normalization check.  For the
counter-propagating rarefaction-wave geometry, the present coefficient reduces
to the differential shock-shape relation of
\citet{MeshkovOmelchenkoUskov2002}.  Differentiation of the exact shock polar
and direct solution of the complete differentiated Rankine--Hugoniot/modal
system agree over 285 finite-strength states, with a maximum relative
discrepancy of \(2.73\times10^{-15}\).  This verifies the sign and physical
normalization used below without adding a separate planar derivation to the
main paper.  The comparison targets a first-jet identity between one-sided
derivatives at a single interaction event.  A grid calculation would recover
the same derivatives through finite stencils and would mix the junction law
with shock-location and truncation errors, whereas differentiation of the
exact Rankine--Hugoniot shock map tests the local identity directly.  A global
computation would address the complementary question of how the local
junction continues within a complete flow field.

The next section restores the full three-dimensional shock surface and shows
how these mode-dependent scalar coefficients become different intrinsic and
downstream flow responses when the interaction curve is embedded in different
ways.

\section{Three-dimensional reconstruction and downstream flow}
\label{sec:three-dimensional-inlet-response}
\label{subsec:physical-incident-normalization}

Section~\ref{sec:steady-euler-curvature-transfer} determines one scalar response,
\begin{equation}
  \kappa=K_{\rm loc}a_{\rm in}.
  \label{eq:inlet-local-scalar-response}
\end{equation}
The active scattering problem is confined to the plane normal to the
interaction curve, but its realization on the physical shock is not.  The
same scalar must be inserted into the full shape operator of the
three-dimensional front.  This is also the quantity needed by smooth
curved-shock reconstruction, in which the one-sided shock geometry enters the
differentiated Rankine--Hugoniot and Euler relations for post-shock pressure,
velocity and vorticity gradients
\citep{Molder2016,Emanuel2018,EmanuelMolder2022,ZhangEtAl2025}.

Let \(\boldsymbol{t}\) be the unit tangent to the instantaneous interaction
curve and let \(\eta_s\) be the orthogonal tangent direction in the shock
surface.  In the Darboux basis \((\boldsymbol{t},\eta_s)\), write
\begin{equation}
  S^-=
  \begin{pmatrix}
    k_{\mathcal C} & \tau_{\mathcal C}\\
    \tau_{\mathcal C} & k_\eta
  \end{pmatrix},
  \qquad
  k_{\mathcal C}=\langle S^-\boldsymbol{t},\boldsymbol{t}\rangle .
  \label{eq:inlet-darboux-data}
\end{equation}
Here \(k_{\mathcal C}\) is the normal curvature of the shock along the
interaction track, \(k_\eta\) is the cross-track normal curvature, and
\(\tau_{\mathcal C}\) is the corresponding geodesic torsion.

\begin{proposition}[Orientation-dependent shape-operator reconstruction]
\label{prop:inlet-three-dimensional-reconstruction}
For a steady \(C^1\), piecewise-\(C^2\) shock junction,
\begin{equation}
  \boxed{
  S^+=S^-+\kappa\,\eta_s\otimes\eta_s^\flat
  =
  \begin{pmatrix}
    k_{\mathcal C} & \tau_{\mathcal C}\\
    \tau_{\mathcal C} & k_\eta+\kappa
  \end{pmatrix}.}
  \label{eq:inlet-shape-operator-update}
\end{equation}
Consequently,
\begin{equation}
  \cjump{H}=\frac{\kappa}{2},
  \qquad
  \boxed{\cjump{K_G}=k_{\mathcal C}\kappa
  =k_{\mathcal C}K_{\rm loc}a_{\rm in}.}
  \label{eq:inlet-mean-gaussian-law}
\end{equation}
For any normal section making an angle \(\theta\) with the interaction track,
\begin{equation}
  \cjump{k_n}=\kappa\sin^2\theta .
  \label{eq:inlet-directional-curvature-law}
\end{equation}
\end{proposition}

The proposition is the steady spatial form of the rank-one junction law in
Section~\ref{sec:local-euler-junction-law}; its proof is obtained by evaluating
that tensor on the Darboux frame and taking trace and determinant.  It has a
simple physical meaning: the weak sheet leaves the along-track normal
curvature and geodesic torsion continuous and deposits curvature only in the
cross-track direction.  If the interaction track is asymptotic,
\(k_{\mathcal C}=0\), Gaussian curvature is preserved even though the shape
operator and mean curvature change.  If the track is not asymptotic, the same
normal-plane response produces a non-zero intrinsic curvature junction.

For dimensional interpretation we use a local length scale \(L_0\) and the
upstream acoustic pressure scale \(\rho_1c_1^2\):
\begin{equation}
  a_{\rm in}=\frac{L_0}{\rho_1c_1^2}
  \sjump{D_{\eta_s}p}_{\Gamma},
  \qquad
  L_0\kappa_{\rm physical}=K_{\rm loc}a_{\rm in}.
  \label{eq:inlet-physical-normalization}
\end{equation}
Thus a weak discontinuity has no pressure jump; \(a_{\rm in}\) measures a
pressure-gradient jump.  This normalization is used below.

To isolate the difference between normal-plane scattering and
three-dimensional embedding, consider the inlet-inspired saddle patch used by
\citet{ZhangEtAl2025},
\begin{equation}
  \Sigma^-:\qquad x-\frac{y^2}{5}+\frac{z^2}{5}=0,
  \qquad M_1=3,
  \label{eq:inlet-saddle-surface}
\end{equation}
and the event
\begin{equation}
  P=\left(\frac35,2,1\right),
  \qquad
  N(P)=\frac{(5,-4,2)}{3\sqrt5},
  \qquad
  K_G^-(P)=-\frac4{81}.
  \label{eq:inlet-example-point}
\end{equation}
Two stationary acoustic sheets are chosen so that their complete active
normal-plane boundary problems coincide, while their intersection tracks have
different embeddings in the same shock surface.  The adapted frames are
\begin{equation}
\begin{aligned}
  \boldsymbol{t}_A&=\frac{(2,5,5)}{3\sqrt6},
  &\eta_A&=\frac{(10,7,-11)}{3\sqrt{30}},\\
  \boldsymbol{t}_B&=\frac{(-2,1,7)}{3\sqrt6},
  &\eta_B&=\frac{(10,13,1)}{3\sqrt{30}}.
\end{aligned}
  \label{eq:inlet-two-darboux-frames}
\end{equation}
The exact acoustic conormals are given in the supplementary material.  In
both cases
\begin{equation}
  \frac{a_1}{c_1}=\sqrt5,
  \qquad
  \frac{b}{c_1}=\sqrt{\frac{10}{3}},
  \qquad
  M_{\Pi2}^2=\frac{35}{17},
  \qquad
  K_{\rm loc}=0.1778567623,
  \label{eq:inlet-common-active-data}
\end{equation}
whereas the passive line-tangent velocities have opposite signs.  The two
tracks are shown in Figure~\ref{fig:inlet-saddle-geometry}; track \(A\) is an
asymptotic ruling and track \(B\) is generic.

\begin{figure}
  \centering
  \includegraphics[width=0.82\columnwidth]{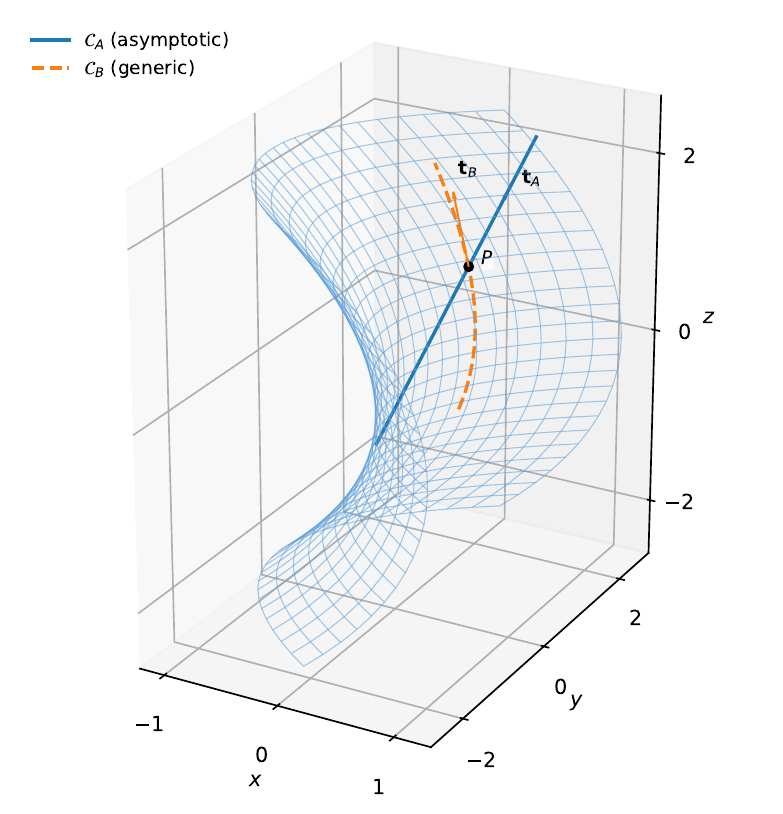}
  \caption{The inlet-inspired saddle shock and the two acoustic intersection
  tracks through \(P\).  The blue track \(\mathcal C_A\) is asymptotic at the
  event, while \(\mathcal C_B\) is generic.  Their active normal-plane Euler
  problems are locally isomorphic.}
  \label{fig:inlet-saddle-geometry}
\end{figure}

The relevant Darboux components are
\begin{equation}
\begin{array}{c|ccc}
  &k_{\mathcal C}&\tau_{\mathcal C}&k_\eta\\ \hline
 A&0&2/9&-8\sqrt5/225\\
 B&-16\sqrt5/135&2/135&56\sqrt5/675.
\end{array}
  \label{eq:inlet-darboux-table}
\end{equation}
Both rows have determinant \(-4/81\).

\begin{theorem}[Identical active scattering, distinct intrinsic junctions]
\label{thm:inlet-controlled-pair}
The two configurations have identical incident active traces, outgoing scalar
amplitudes, \(K_{\rm loc}\), and \(\kappa\), but
\begin{equation}
  K_{G,A}^+=-\frac4{81},
  \qquad
  K_{G,B}^+=-\frac4{81}-\frac{16\sqrt5}{135}\,\kappa .
  \label{eq:inlet-two-gaussian-responses}
\end{equation}
Thus the complete active normal-plane Euler problem does not determine the
intrinsic response of the three-dimensional shock.
\end{theorem}

This comparison is deliberately controlled rather than representative of a
complete inlet solution: it holds the finite shock strength and all active
scattering data fixed and changes only the embedding of the interaction
track.  The same construction applies to acoustic, entropy and in-plane
vortical incidence because their active boundary problems are identical in
frames \(A\) and \(B\); line-tangent vorticity remains curvature-transparent.

The geometric difference propagates directly into one-sided flow gradients.
For either smooth side \(\varepsilon\in\{-,+\}\) and each tangent direction
\(X\), differentiated Rankine--Hugoniot conditions and the steady Euler
equations give
\begin{equation}
\begin{split}
  \mathcal A_N(Y_2)D_XY_2^\varepsilon
  &=\mathcal A_N(Y_1)D_XY_1^\varepsilon
    +\sjump{F_{S^\varepsilon X}}_\Sigma,\\
  \mathcal A_N(Y_2)D_NY_2^\varepsilon
  &+\mathcal A_{\eta_s}(Y_2)D_{\eta_s}Y_2^\varepsilon
   +\mathcal A_{\boldsymbol{t}}(Y_2)D_{\boldsymbol{t}}Y_2^\varepsilon=0.
\end{split}
  \label{eq:inlet-one-sided-reconstruction}
\end{equation}
The first line inserts the one-sided shape operator; the second reconstructs
the normal derivative.  Define
\begin{equation}
  \widehat{\nabla p_2}=\frac{L_0}{\rho_1c_1^2}\nabla p_2,
  \qquad
  \widehat{\boldsymbol{\omega}}_2
  =\frac{L_0}{c_1}\nabla\times\boldsymbol{u}_2.
  \label{eq:inlet-gradient-vorticity-normalization}
\end{equation}
For unit acoustic incidence the pre-interaction vectors reconstructed in the
two frames agree in the fixed laboratory basis, providing a coordinate check.
Across the junction,
\begin{equation}
\begin{split}
  \cjump{\widehat{\nabla p_2}}^{(j)}
  &=4.366454\,
    \bigl(-0.400525N+0.916286\eta_j\bigr),\\
  \cjump{\widehat{\boldsymbol{\omega}}_2}^{(j)}
  &=-0.595955\,\boldsymbol{t}_j,
  \qquad j=A,B.
\end{split}
  \label{eq:inlet-gradient-vorticity-jumps}
\end{equation}
The coefficients are the common outgoing acoustic and in-plane shear
amplitudes; only their physical directions differ.

\begin{table}
  \caption{Post-interaction downstream observables for unit acoustic
  incidence.  The angle is between the two laboratory-frame vectors in cases
  \(A\) and \(B\).  Full vector components are given in the supplementary
  material.}
  \label{tab:inlet-postshock-observables}
  \centering
  \begin{tabular}{lccc}
    observable & \(\lvert\cdot\rvert_A\) & \(\lvert\cdot\rvert_B\) & angle\\
    \hline
    \(\widehat{\nabla p_2^+}\) & 6.2272 & 5.4882 & \(31.59^\circ\)\\
    \(\widehat{\boldsymbol{\omega}}_2^+\) & 0.8413 & 1.0813 & \(25.64^\circ\)
  \end{tabular}
\end{table}

\begin{proposition}[Identical active scattering, distinct downstream flow]
\label{prop:inlet-downstream-observables}
Cases \(A\) and \(B\) have identical active normal-plane scattering and the
same pre-interaction laboratory-frame pressure-gradient and vorticity vectors,
but their post-interaction vectors have the distinct magnitudes and directions
listed in Table~\ref{tab:inlet-postshock-observables}.  Hence the embedding of
an otherwise identical modal response changes the three-dimensional
post-shock gradient and rotational structure.
\end{proposition}

For entropy incidence the corresponding pressure-gradient and vorticity
angles are \(13.35^\circ\) and \(51.62^\circ\); for in-plane vortical incidence
they are \(8.77^\circ\) and \(44.38^\circ\).  The complete vectors, exact
acoustic-plane construction, Gaussian-curvature plot and residual checks are
reported in the supplementary material.  Central differences of the exact
perfect-gas shock map give a maximum discrepancy of \(9.2\times10^{-10}\)
in standard double precision and \(8.0\times10^{-27}\) in 80-digit
arithmetic with step \(10^{-13}\); the modal and steady-Euler residuals
remain below \(2.3\times10^{-15}\) in the double-precision
implementation.  The controlled pair therefore exhibits a genuine
downstream gas-dynamic distinction, not merely a difference of surface
invariants.

\section{Moving spherical shocks and near-umbilic response}
\label{sec:moving-interaction-curves}

The preceding sections determined the scalar junction amplitude at a frozen
interaction event and reconstructed its three-dimensional spatial effect.  We
now allow both the shock and the intersection curve to move.  The essential
point is that the same spacetime rank-one tensor simultaneously controls the
jump of shock curvature, the transverse gradient of normal shock speed and
the intrinsic normal-acceleration trace.  A constant inertial Galilean boost
then reduces the local Euler coefficient to the steady formula evaluated at
relative velocities.  The expanding spherical example below has non-zero
shock speed and non-zero tracking speed and therefore tests both parts of the
construction.  Because a sphere is exactly umbilic, it also provides the
natural bridge to the compact near-umbilic result at the end of the section.

Let
\begin{equation}
  \nu=N^\flat-V_n\,dt,
  \qquad
  \mathsf T=\partial_t+V_nN,
  \qquad
  \mathsf Z=\mathsf T+q\eta_s,
  \label{eq:moving-gauge-tracking}
\end{equation}
where \(V_n\) is the normal shock speed and \(q\) is the signed velocity with
which the interaction curve moves across the shock in the \(\eta_s\)
direction.  The quantity
\begin{equation}
  \frac{\delta V_n}{\delta t}:=\mathsf T(V_n)
  \label{eq:moving-intrinsic-acceleration}
\end{equation}
is the intrinsic normal time derivative of the shock speed.  The covector on
\(T\Sigma\) transverse to the interaction world sheet is
\begin{equation}
  \vartheta=\eta_s^\flat-q\,dt.
  \label{eq:moving-transverse-covector}
\end{equation}

\begin{proposition}[Moving-junction compatibility]
\label{prop:moving-spacetime-compatibility}
For a moving interaction curve,
\begin{equation}
  \cjump{\mathcal B_\nu}
  =\kappa(\eta_s^\flat-q\,dt)
          \otimes(\eta_s^\flat-q\,dt).
  \label{eq:moving-rank-one-law}
\end{equation}
Consequently,
\begin{equation}
  \boxed{
  \begin{aligned}
  \cjump{S}&=\kappa\,\eta_s\otimes\eta_s^\flat,\\
  \cjump{\nabla_{\Sigma_t}V_n}&=-q\kappa\,\eta_s,\\
  \cjump{\dfrac{\delta V_n}{\delta t}}&=q^2\kappa.
  \end{aligned}}
  \label{eq:moving-curvature-speed-acceleration-law}
\end{equation}
The three non-zero traces satisfy
\begin{equation}
  \cjump{\frac{\delta V_n}{\delta t}}\,\cjump{k_\eta}
  -\cjump{\eta_s(V_n)}^2=0,
  \label{eq:moving-rank-one-minor}
\end{equation}
and, together with \(\cjump{K_G}=k_{\mathcal C}\kappa\),
\begin{equation}
  q^2\cjump{K_G}
  =k_{\mathcal C}\cjump{\frac{\delta V_n}{\delta t}}.
  \label{eq:moving-gaussian-acceleration-coupling}
\end{equation}
\end{proposition}

Equation~\eqref{eq:moving-curvature-speed-acceleration-law} determines jumps
between the two smooth shock pieces, not the total acceleration on either
side.  The one-sided pressure and velocity gradients determine the smooth
values of \(V_n\) and \(\delta V_n/\delta t\); the junction law fixes how
those traces differ when the weak sheet passes.  The proof is the direct
expansion of the rank-one law of Proposition~\ref{prop:spacetime-rank-one};
the componentwise conormal calculation is given in the supplementary
material.

To compute \(\kappa\), freeze one interaction event and use the constant
inertial velocity
\begin{equation}
  \boldsymbol c_p=V_nN+q\eta_s.
  \label{eq:moving-eventwise-boost}
\end{equation}
The corresponding Galilean transformation makes the shock and the
interaction line instantaneously stationary and replaces the active velocity
components by
\begin{equation}
  \bar a_i=a_i-V_n,
  \qquad
  \bar b=b-q,
  \qquad
  \bar w_i=w_i.
  \label{eq:moving-relative-components}
\end{equation}

\begin{lemma}[Eventwise inertial Galilean reduction]
\label{lem:eventwise-galilean-reduction}
At each interaction event, the complete differentiated Euler boundary
operator is conjugate to the stationary operator evaluated at the barred
traces.  Hence
\begin{equation}
  \boxed{
  K_{\rm loc}^{\rm moving}
  =K_{\rm loc}^{\rm steady}
   (\bar Y_1,\bar Y_2;\bar q_{\rm in}).}
  \label{eq:moving-galilean-coefficient}
\end{equation}
\end{lemma}

The boost in Lemma~\ref{lem:eventwise-galilean-reduction} is constant and is
chosen separately at each event; it is not a rotating or accelerating frame
attached to the curved shock.  The exact flux-conjugacy proof and the
associated regression test are given in the supplementary material.  If the
weak-sheet conormal is
\begin{equation}
  \zeta=\zeta_0dt+m_NN^\flat+m_\eta\eta_s^\flat,
\end{equation}
then its tracking velocity is
\begin{equation}
  q=-\frac{\zeta_0+m_NV_n}{m_\eta}.
  \label{eq:moving-tracking-speed}
\end{equation}
The scalar coefficient of Section~\ref{sec:steady-euler-curvature-transfer}
applies when
\begin{equation}
  \bar a_1>c_1,
  \qquad
  0<\bar a_2<c_2,
  \qquad
  \bar a_2^2+\bar b^2>c_2^2.
  \label{eq:moving-mode-gate}
\end{equation}
The last inequality is a relative projected-supersonic condition.  Its
boundary changes the outgoing mode count; it is not a zero of the static
transversality pairing.

Consider now an expanding spherical shock
\begin{equation}
  \Sigma_t=\{x:|x|=R(t)\},
  \qquad
  D(t)=\dot R(t)>0,
  \label{eq:blast-spherical-front}
\end{equation}
propagating into a quiescent perfect gas.  Orient the normal from the
undisturbed exterior gas toward the shocked interior, so that
\begin{equation}
  N=-\boldsymbol e_r,
  \qquad
  V_n=-D,
  \qquad
  S^-=\frac1R I.
  \label{eq:blast-normal-speed-shape}
\end{equation}
This is a genuinely moving, exactly umbilic shock.  Such spherical shocks are
a canonical unsteady setting for shock-acceleration discontinuities
\citep{Taylor1950,Sedov1959,VirgopiaFerraioli1982,RadhaSharmaJeffrey1993}.

At an event \((t_0,P)\), set \(R_0=R(t_0)\), take
\begin{equation}
  \frac{D(t_0)}{c_1}=5,
  \qquad
  \gamma=\frac75,
\end{equation}
and let the incident acoustic sheet have spatial conormal
\begin{equation}
  \boldsymbol m=\frac35N-\frac45\eta_s.
  \label{eq:blast-acoustic-sheet}
\end{equation}
Equation~\eqref{eq:moving-tracking-speed} gives
\begin{equation}
  \frac{V_n}{c_1}=-5,
  \qquad
  \frac{q}{c_1}=-\frac52,
  \qquad
  \frac{\bar a_1}{c_1}=5,
  \qquad
  \frac{\bar b}{c_1}=\frac52.
  \label{eq:blast-relative-state}
\end{equation}
The downstream state satisfies
\begin{equation}
  \frac{\rho_2}{\rho_1}=5,
  \qquad
  \frac{\bar a_2}{c_1}=1,
  \qquad
  \frac{c_2^2}{c_1^2}=\frac{29}{5},
  \qquad
  \bar M_{\Pi2}^2=\frac54.
  \label{eq:blast-downstream-state}
\end{equation}
Here
\begin{equation}
  B_{\rm crit}=\sqrt{R_{7/5}(25)}=\sqrt{\frac{24}{5}}
  =2.190890\ldots,
  \qquad
  \frac{B}{B_{\rm crit}}=1.14109\ldots,
  \label{eq:blast-sonic-margin}
\end{equation}
so the event lies inside the scalar mode sector with a finite but moderate
\(14.1\%\) margin in relative tangential speed.

For pressure-gradient-normalized acoustic incidence, the coefficient is
entirely rational:
\begin{equation}
  \Delta_{\rm loc}=\frac{14}{29},
  \qquad
  \omega_\chi(\widetilde q_{\rm in})=\frac{48}{145},
  \qquad
  \boxed{K_{\rm loc}^{\rm blast}=\frac{24}{35}.}
  \label{eq:blast-exact-coefficient}
\end{equation}
By Lemma~\ref{lem:eventwise-galilean-reduction}, the same coefficient is
obtained for the frozen stationary problem with relative upstream components
\((\bar a_1,\bar b)=(5,5/2)c_1\), although the laboratory event has
\(V_n=-5c_1\).  The sphere therefore provides a physical realization of the
eventwise Galilean reduction.

With \(L_0=R_0\), write
\(\widehat\kappa=R_0\kappa_{\rm physical}\).  Then
\begin{equation}
  \boxed{
  \widehat\kappa=\frac{24}{35}a_{\rm in},
  \qquad
  R_0^2\cjump{K_G}=\frac{24}{35}a_{\rm in},}
  \label{eq:blast-curvature-response}
\end{equation}
while the kinematic traces are
\begin{equation}
  \boxed{
  \frac{R_0}{c_1}\cjump{\eta_s(V_n)}=\frac{12}{7}a_{\rm in},
  \qquad
  \frac{R_0}{c_1^2}
  \cjump{\frac{\delta V_n}{\delta t}}
  =\frac{30}{7}a_{\rm in}.}
  \label{eq:blast-speed-acceleration-jumps}
\end{equation}
If
\begin{equation}
  A_0=\left(\frac{\delta V_n}{\delta t}\right)^-
\end{equation}
is the smooth one-sided blast-wave acceleration, then
\begin{equation}
  \left(\frac{\delta V_n}{\delta t}\right)^+
  =A_0+\frac{30}{7}\frac{c_1^2}{R_0}a_{\rm in}.
  \label{eq:blast-total-versus-jump}
\end{equation}
The theory therefore adds a precisely determined acceleration discontinuity
to an otherwise arbitrary smooth blast-wave evolution.  At the same time,
the initially repeated curvature \(1/R_0\) splits: the \(\boldsymbol t\)
curvature remains \(1/R_0\), whereas the \(\eta_s\)-curvature becomes
\((1+\widehat\kappa)/R_0\).  The moving spherical shock thus realizes both
the classical scalar acceleration jump and the additional multidimensional
shape-operator response.

The exact sphere corresponds to an umbilic.  The local response near, but not
exactly at, an umbilic is summarized by the following compact result.  In a
principal tangent frame let
\begin{equation}
  S^-=\begin{pmatrix}k+\delta&0\\0&k-\delta\end{pmatrix},
  \qquad
  \eta_s=\cos\alpha\,e_1+\sin\alpha\,e_2,
  \qquad
  |\delta|<|k|.
  \label{eq:near-umbilic-base}
\end{equation}
Define the complex curvature anisotropy
\begin{equation}
  \mathfrak a(S)=(S_{11}-S_{22})+2iS_{12}.
\end{equation}

\begin{proposition}[Umbilic splitting and near-umbilic frame rotation]
\label{prop:bow-umbilic-splitting}
The junction law gives
\begin{equation}
  \boxed{\mathfrak a(S^+)=2\delta+\kappa e^{2i\alpha}.}
  \label{eq:bow-anisotropy-addition}
\end{equation}
Hence
\begin{equation}
  |k_1^+-k_2^+|=|2\delta+\kappa e^{2i\alpha}|,
  \qquad
  2\Delta\psi=\arg\!\left(1+\frac{\kappa}{2\delta}
  e^{2i\alpha}\right)
  \quad(\bmod\ \pi)
  \label{eq:near-umbilic-gap-angle}
\end{equation}
for \(\delta\neq0\).  At an exact umbilic \((\delta=0)\), every non-zero
\(\kappa\) selects \((\eta_s,\boldsymbol t)\) as the post-interaction
principal frame, with principal curvatures \(k+\kappa\) and \(k\).  If
\(r=\kappa/(2\delta)\) and \(|r|<1\), then
\begin{equation}
  \boxed{|\Delta\psi|\leq\frac12\arcsin|r|,}
  \qquad
  \cos2\alpha=-r
  \quad\text{at equality}.
  \label{eq:bow-sharp-rotation-bound}
\end{equation}
For a near-umbilic family with \(\kappa=2\lambda\delta\) and
\(|\delta|/|k|\to0\), the ratio \(|\kappa|/|k|\to0\) while
\(\Delta\psi\) can remain finite.
\end{proposition}

Equation~\eqref{eq:bow-anisotropy-addition} follows by inserting
\(S^+=S^-+\kappa\eta_s\otimes\eta_s\) in the principal frame.  The sharp
bound is the tangent angle subtended at the origin by the circle
\(1+r e^{2i\alpha}\); full details, the response diagram and the
created-umbilic condition are given in the supplementary material.  Thus the
Euler coefficient need not become large: an order-one rotation occurs because
the base principal-curvature gap becomes small.  The exact spherical example
shows how the same rank-one junction selects a frame on a moving umbilic
shock, while Proposition~\ref{prop:bow-umbilic-splitting} quantifies the
corresponding sensitivity near an umbilic.

\section{Conclusions}
\label{sec:conclusions}

A weak characteristic sheet crossing a finite-strength shock selects one
scalar junction amplitude through the differentiated Euler boundary problem.
The geometric content of the result is that this scalar is not itself the
three-dimensional response: it is deposited as a rank-one update of the full
shape operator,
\begin{equation}
  \kappa=K_{\rm loc}a_{\rm in},
  \qquad
  \cjump{S}=\kappa\,\eta_s\otimes\eta_s^\flat,
  \qquad
  \cjump{K_G}=k_{\mathcal C}\kappa .
  \label{eq:conclusion-junction-law}
\end{equation}
The local gas dynamics fixes \(\kappa\), while the embedding of the
interaction curve fixes how that amplitude changes the intrinsic shock
geometry and the one-sided flow reconstruction.  In this sense the result
extends the classical scalar acceleration-discontinuity description to a
junction law for a continuously differentiable, piecewise twice-differentiable
three-dimensional front.

The controlled saddle comparison isolates the resulting orientation effect.
Two interactions with identical active normal-plane states, incident traces
and outgoing modal amplitudes have the same \(K_{\rm loc}\) and the same
\(\kappa\), but different Gaussian-curvature changes.  More importantly, the
one-sided Euler reconstruction yields different laboratory-frame downstream
fields: the post-interaction pressure-gradient directions differ by
\(31.6^\circ\) and the vorticity directions by \(25.6^\circ\).  Thus the
three-dimensional embedding changes observable post-shock gradient and
rotational structure even when the complete active scattering problem is
unchanged.

The perfect-gas coefficients identify two further selection mechanisms.
Acoustic incidence has curvature-neutral branches on which finite outgoing
acoustic, entropy and shear waves coexist with zero deposited curvature.
Convected entropy and in-plane vortical sheets deposit curvature with opposite
signs and satisfy
\(K_{\rm vort}=-2K_{\rm ent}/M_{\Pi1}\), a consequence of total-enthalpy
conservation and acoustic orthogonality; line-tangent vorticity remains in the
passive channel.  These results connect the junction law to elementary
components of shock--turbulence interaction rather than to acoustic forcing
alone.

For a moving interaction curve, the same spacetime compatibility condition
also fixes the jumps of transverse shock-speed gradient and intrinsic normal
acceleration.  The expanding spherical example has non-zero shock speed and
non-zero tracking speed, and separates the smooth blast-wave acceleration
from the additional deposited acceleration jump.  At an exact umbilic the
rank-one update selects a principal frame; near an umbilic it can rotate that
frame by an order-one angle when the response is comparable with the small
principal-curvature gap.  These are local first-jet results.  Their role is to
supply interface data for one-sided curved-shock reconstruction and for
shock-fitting or front-tracking calculations; global continuation, wall and
boundary-layer coupling, and subsequent wave transport remain separate
problems.

\begin{bmhead}[Supplementary material.]
Supplementary material is supplied with this paper.  It contains the detailed
spacetime compatibility and Galilean reductions, the full swept-shock
mode-count and perfect-gas shock-polar algebra, modal and downstream-vector
tables, additional moving and near-umbilic benchmarks, the complete planar
shock-polar verification, and reproducibility information for all numerical
coefficients and figures.  The accompanying source archive contains the
scripts and tabulated data used in these calculations.
\end{bmhead}

\begin{bmhead}[Funding.]
This research received no specific grant from any funding agency, commercial
or not-for-profit sectors.
\end{bmhead}

\begin{bmhead}[Declaration of interests.]
The author reports no conflict of interest.
\end{bmhead}

\begin{bmhead}[Data availability statement.]
All data needed to reproduce the reported numerical coefficients, tables and
figures are provided in the supplementary material and the accompanying
source-code and data archive.
\end{bmhead}

\begin{bmhead}[Declaration of use of artificial intelligence (AI).]
OpenAI's ChatGPT was used as an interactive
research and writing assistant to support manuscript organisation and
language editing, explore alternative presentations of author-developed
mathematical arguments, and assist in preparing and reviewing numerical
verification code. Anthropic's Claude was used for critical reading and
editorial suggestions. All mathematical results, derivations, numerical calculations, references,
physical interpretations and final wording were reviewed, revised and
approved by the author. No confidential or unpublished third-party material was
supplied to the tools. The author assumes full responsibility for the
accuracy, integrity and originality of the manuscript.
\end{bmhead}

\begin{bmhead}[Author contributions.]
Alexander Omelchenko: Conceptualization, Formal analysis, Investigation,
Methodology, Software, Validation, Visualization, Writing--original draft, and
Writing--review and editing.
\end{bmhead}

\bibliographystyle{jfm}
\bibliography{refs}

\end{document}